\documentclass[11pt]{article}

\usepackage[latin1]{inputenc}
\usepackage[T1]{fontenc}
\usepackage[english]{babel}
\usepackage{amsfonts}
\usepackage{amssymb}
\usepackage{amsmath}
\usepackage{amsthm}
\usepackage{graphicx}

\newtheorem{proposizione}{Proposition}
\newtheorem{lemma}{Lemma}
\newtheorem{teorema}{Theorem}

\def\be{\begin{equation}}
\def\ee{\end{equation}}
\def\bea{\begin{eqnarray}}
\def\eea{\end{eqnarray}}
\def\nin{\noindent}
\def\nn{\nonumber}

\newcommand{\meanv}[1]{\left\langle#1\right\rangle}
\newcommand{\R}{\mathbb{R}}

\newcommand{\E}{\mathbb{E}}

\newcommand{\epsi}{\varepsilon}

\newcommand{\Or}[1]{O \left(#1\right)}

\def\s{\sigma}

\def\b{\beta}

\def\e{\eta}

\def\l{\lambda}

\begin{document}

\title{A Solvable Mean Field Model\\ of a Gaussian Spin Glass}


\author{Adriano Barra\footnote{Dipartimento di Fisica, Sapienza Universit\`a di Roma and GNFM, Gruppo di Roma1.},
Giuseppe Genovese\footnote{Dipartimento di  Matematica,
    Sapienza Universit\`a di Roma.}, \ Francesco Guerra\footnote{Dipartimento di Fisica, Sapienza Universit\`a di Roma and INFN, Sezione di Roma.}, Daniele Tantari\footnote{Dipartimento di Fisica, Sapienza Universit\`a di Roma.}}

\date{\today}

\maketitle

\begin{abstract}
\nin We introduce a mean field spin glass model with gaussian distribuited spins and pairwise interactions, 
whose couplings are drawn randomly from a gaussian distribution $\mathcal{N}(0,1)$ too. We completely control
 the main thermodynamical properties of the model (free energy, phase diagram, fluctuations theory) in the 
whole phase space. In particular we prove 
that in thermodynamic limit the free energy equals its replica symmetric expression.
\end{abstract}

\section*{Introduction}
Recently, some work has been done studying the properties of bipartite spin glasses \cite{HT}\cite{NN}\cite{Bip}. The main interest in these models is related to the peculiarity of the Hopfield Model, a well known model of very hard solution from a mathematical point of view (see \cite{T} and references therein), can be seen as a special bipartite model, with a party of \textit{usual} dichotomic spin, and another party of \textit{special} gaussian soft spin variables.

\nin In particular, from the investigation of dichotomic bipartite spin glasses, it has been shown that, at least to the Replica Symmetric approssimation (with zero external field), the model can be written as a convex combination of two different Sherrington-Kirkpatrick models, at different temperatures \cite{Bip}. This seems to be more than a hint that a similar structure should be conserved in the Hopfield Model, and infact we have recently shown that this is the case \cite{NN?}. 

\nin As a consequence, while the dichotomic spin glass has been intensively studied, the need for a clear picture of its gaussian counterpart is the main interest for this paper.

\nin The Gaussian Spin Model has been originally introduced together with the Spherical Model in \cite{BK}. Then, due to the natural divergences arising in such a model, the main interest has been concentrated in the Spherical one \cite{KT}\cite{CS}\cite{Tsfer}\cite{PT}\cite{T}, tought in some recent papers are discussed interesting properties of gaussian models similar to the one we introduce here for the first time \cite{BDG}\cite{BoK}\cite{FS}.

\nin Therefore in our work we extend here techniques previously developed for pairwise dichotomic spin glasses to their gaussian counterparts. 

\nin In Section 1 we introduce the model with all its related statistical mechanics package and regularize it so to avoid divergencies due to coupled fat tails of the soft (unbounded) spin. 

\nin In Section 2, we show how to get a rigorous control of the thermodynamic limit of the free energy. 

\nin Section 3 is left for the investigation of the high temperature limit (the ergodic behavior). 

\nin In Section 4 we develop a generalization of a sum rule for the free energy in terms of its replica symmetric approximation and an error term. The breaking of ergodicity is expected to be a critical phenomenon. 

\nin Section 5 is dedicated to a fluctuation theory for the order parameter.

\nin In Section 6 we develop the broken replica symmetric bound, coupled with the Parisi-like equation, and we show that it is equivalent to the bound given by replica symmetric solution.

\nin In Section 7 we finally prove a lower bound for the free energy, stating that the correct solution is the RS expression.

\section{Definition of the Model}
 We introduce a system on $N$ sites, whose generic configuration is defined by spin variables $z_i\in \R$, $i=1,2,...,N$ attached on each site. We call the
 external quenched disorder a set of $N^2$ indipendent and identical distributed random variables $J_{ij}$, defined for each couple of sites $(i,j)$.
 We assume each $J_{ij}$ to be a centered unit Gaussian $\mathcal{N}(0,1)$ i.e. $$\E(J_{ij})=0, \ \ \ \  \E(J^2_{ij})=1.$$ The interaction among spins is given
by  defining the Hamiltonian $$H_N(z,J)=-\frac{1}{\sqrt{2N}}\sum_{i,j}^NJ_{ij}z_iz_j-h\sum_{i=1}^Nz_i.$$ The first sum, extending to all spin couples, with the factor
 $1/\sqrt{N}$, is the typical long range spin-spin interaction of the mean field spin glass model. The second sum, extending to all sites, is the one-body
 interaction with a scalar external field $h \in \mathbb{R}$. All the thermodynamic properties of the model are codified in the partition function that we write symbolically as
 $$Z_N(\beta,J)=\sum_{configurations}e^{-\beta H_N(z,J)},$$for a given inverse temperature $\beta$.
 In our model we state that there are a number of identical $z$-type configurations proportional to $d\mu(z)$, where $d\mu (z)=d\mu (z_1)...d\mu (z_N)$,
 $d\mu(z_i)=(2\pi)^{-\frac{1}{2}}\exp(-z_i^2/2)$, so to justify the following definition
 \begin{equation}\label{e1}
 Z_N(\beta,J)=\int d\mu (z)e^{-\beta H_N(z,J)}=\E_ze^{-\beta H_N(z,J)} .
 \end{equation}
 \nin Substantially, we called this kind of model "fully gaussian spin glass" because the
 external quenched disorder as well as the value of the soft-spin variables are drawn from a Gaussian distribution  $\mathcal{N}(0,1)$. As early pointed out for instance in \cite{BK}, unfortunately these kind of models need to be regularized; in fact, the right side of
 $(\ref{e1})$ is not always well defined as the pairwise interaction bridges soft spins which are both Gaussian distributed.
 \newline
 It will be clear soon that a good definition is
 \be\label{e2}
 Z_N(\beta,J,\lambda)=\E_z\exp\big[-\beta H_N(z,J)-\frac{\beta^2}{4N}(\sum_{i=1}^N z_i^2)^2+\frac{\lambda}{2}\sum_{i=1}^Nz_i^2\big],
 \ee
 where the first additional term is needed for convergence of the integral over the Gaussian measure $\mu(z)$ as it essentially flattens the Gaussian tails of the variables $z_i$. The new parameter $\lambda$, within the
 last term of ($\ref{e2}$), instead is inserted just to modify the variance of the soft spins, as in several applications this can sensibly  vary. \\
 For a given inverse temperature (or noise level) $\beta$, we introduce the (quenched average of the) free energy per site $f_N(\beta)$, the Boltzmann state $\omega_J$ and the auxiliary function $A_N(\b)$ (namely the pressure), according to the definition
 \begin{equation}\label{e3}
 -\beta f_N(\beta)=A_N(\b)=N^{-1}\E \log Z_N(\beta,J),
 \end{equation}
 \begin{equation}
 \omega_J(O)=Z_N^{-1} \E_z O(z)\exp\big[-\beta H_N(z,J)-\frac{\beta^2}{4N}(\sum_{i=1}^N z_i^2)^2+\frac{\lambda}{2}\sum_{i=1}^Nz_i^2\big],
 \end{equation}
 where $O$ is a generic function of the $z$'s. In the notation $\omega_J$, we have stressed the dependence of the Boltzmann state on the external noise $J$, but,   of course, there is also a dependence on $\beta$,$h$ and $N$.\\ Let us now introduce the important concept of replicas. Consider a generic number $s$ of
 independent copies of the system, characterized by the spin variables $z^{(1)}_i,\ldots z^{(s)}_i$ distributed according to the product state
 \be \Omega_J=\omega^{(1)}_J\ldots\omega^{(s)}_J, \ee
 where all $\omega^{(\alpha)}_J$ act on each one $z_i^{(\alpha)}$'s, and are subject to the same sample $J$ of the external noise. Finally, for a generic smooth
 function $F(z^{(1)}_i,\ldots z^{(s)}_i)$ of the replicated spin variables, we define the $\meanv{.}$  average as
 \be\label{e4} \meanv{F(z^{(1)}_i,\ldots z^{(s)}_i)}=\E\Omega_J(F(z^{(1)}_i,\ldots z^{(s)}_i)). \ee
 Correlation functions are also well defined as overlap $q$ between replicas:
 $$
 q_{ab,N}=\frac{1}{N}\sum_{i=1}^{N}z^a_iz_i^b.
 $$
 Note that, once defined the overlap among replicas we can write
 \be
 Z_N(\b,\l,J)=\E_z\exp\left(-\beta\sqrt{\frac N 2} \mathcal{K}(z)-\frac{\beta^2}{4N}(\sum_{i=1}^N z_i^2)^2+\frac{\lambda}{2}\sum_{i=1}^Nz_i^2\right),
 \ee 
 where $\mathcal{K}(z)$ is a family of centered gaussian random variables with covariances $S_{zz'}=\E[\mathcal{K}(z)\mathcal{K}(z')]=q_{zz'}^2$ and the regularization term is just $\frac 1 2 \frac{\b^2 N}{2} q_{zz}^2=\frac 1 2 \frac{\b^2 N}{2} S_{zz}$.   

 \section{Thermodynamic Limit}
 The aim of this section is to show how to get a rigorous control of the infinite volume limit of the free energy $f_N$ (or similarly $A_N$). The main idea, inspired by \cite{GT}, is to         compare $A_N$, $A_{N_1}$ and $A_{N_2}$, with $N=N_1+N_2$. For this purpose we consider both  the original $N$
 site system and two independent subsystems made of by $N_1$ and $N_2$ soft spins respectively, so to define
 \bea \label{e5}
 Z_N(t)=&\E_z&\exp\left( \beta\sqrt{\frac{t}{2N}}\sum_{i,j=1}^N J_{ij}z_iz_j -t\frac{\beta^2}{4N}(\sum_{i=1}^N z_i^2)^2\right)\nn\\
       & &\exp\left( \beta\sqrt{\frac{1-t}{2N_1}}\sum_{i,j=1}^{N_1} J'_{ij}z_iz_j -(1-t)\frac{\beta^2}{4N_1}(\sum_{i=1}^{N_1} z_i^2)^2\right)\nn\\
       & &\exp\left( \beta\sqrt{\frac{1-t}{2N_2}}\sum_{i,j=N_1+1}^{N} J''_{ij}z_iz_j -(1-t)\frac{\beta^2}{4N_2}(\sum_{i=N_1+1}^{N} z_i^2)^2\right)\nn\\
       & &\exp\left(\beta \label{} h\sum_{i=1}^Nz_i\right)\exp\left(\frac{\lambda}{2}\sum_{i=1}^Nz_i^2\right),
 \eea
 with $0\leq t\leq 1$. The partition function $Z_N(t)$ interpolates between the original N-spin model (obtained for $t=1$) and the two subsystems (of sizes $N_1$ and $N_2$, obtained for $t=0$) equipped with independent noises $J'$ and $J''$, both independent of $J$, \textit{i.e.}
 \bea
 Z_N(1)&=&Z_N(\beta,J,h)\\
 Z_N(0)&=&Z_{N_1}(\beta,J',h)Z_{N_2}(\beta,J'',h).
 \eea
 As a consequence, if we define the interpolating function
 \begin{equation}
 \varphi(t)=\frac{1}{N}\E \log Z_N(t),
 \end{equation}
 taking into account the definition ($\ref{e3}$), we have
 \bea\label{e9}
 \varphi(1)&=&A_N(\beta,h)\nn,\\
 \varphi(0)&=&\frac{N_1}{N}A_{N_1}(\beta,h)+\frac{N_2}{N}A_{N_2}(\beta,h).
 \eea
 If we derive $\varphi (t)$ we obtain
 \bea\label {e11}
 \frac{d}{dt}\varphi(t)&&=\frac{1}{N}\E\left(\frac{\beta}{2\sqrt{2tN}}\sum_{i,j=1}^NJ_{ij}\omega_t(z_iz_j)\right)
 -\left\langle\frac{\beta^2}{4N^2}(\sum_{i=1}^N z_i^2)^2 \right\rangle\nn\\
 &&-\frac{1}{N}\E\left(\frac{\beta}{2\sqrt{2(1-t)N_1}}\sum_{i,j=1}^{N_1}J'_{ij}\omega_t(z_iz_j)\right)\nn\\
 &&-\frac{1}{N}\E\left(\frac{\beta}{2\sqrt{2(1-t)N_2}}\sum_{i,j=N_1+1}^{N}J''_{ij}\omega_t(z_iz_j)\right)\nn\\
 &&+\left\langle\frac{\beta^2}{4NN_1}(\sum_{i=1}^{N_1} z_i^2)^2 \right\rangle+\left\langle\frac{\beta^2}{4NN_2}(\sum_{i=N_1+1}^{N} z_i^2)^2 \right\rangle.
 \eea
 with $\omega_t(.)$ and the $\meanv{.}$ average as defined in ($\ref{e3}$)($\ref{e4}$) but corresponding to the \textit{Boltzmannfaktor} coupled to $Z_N(t)$ .
 \newline
 Let us now evaluate for example the first term: using a standard integration by parts on the external noise (Wick's theorem) we obtain
 \bea
 \frac{1}{N}\E\left(\frac{\beta}{2\sqrt{2tN}}\sum_{i,j=1}^{N}J_{ij}\omega_t(z_iz_j)\right)=
 \frac{\beta}{2N\sqrt{2tN}}\sum_{i,j=1}^{N}\E\left(\frac{\partial}{\partial_{J_{ij}}}\omega_t(z_iz_j)\right)\nn,
 \eea
and then we get
 \bea\label{ew}
 &&\frac{\beta}{2N\sqrt{2tN}}\sum_{i,j=1}^{N}\E\left(\frac{\partial}{\partial_{J_{ij}}}\omega_t(z_iz_j)\right)=\nn\\
 &=&\frac{\beta^2}{4N^2}\sum_{i,j=1}^{N}\left\langle z_i^2z_j^2\right\rangle
 -\frac{\beta^2}{4N^2}\sum_{i,j=1}^{N}\left\langle z_i^{(1)}z_j^{(1)}z_i^{(2)}z_j^{(2)}\right\rangle\nn\\
 &=&\left\langle\frac{\beta^2}{4N^2}(\sum_{i=1}^N z_i^2)^2 \right\rangle-\frac{\b^2}{4}\meanv{q_{12,N}^2}.
 \eea
 The other terms of  ($\ref{e11}$) can be evaluated in the same way, therefore
 \be\label{e6}
 \frac{d}{dt}\varphi (t)=-\frac{\beta^2}{4}\left(\left\langle q_{12,N}^2\right\rangle-\frac{N_1}{N}\left\langle q_{12,N_1}^2\right\rangle
 -\frac{N_2}{N}\left\langle q_{12,N_2}^2\right\rangle\right)\\,
 \ee
 where we defined the overlaps
 \bea
 q_{12,N_1}&=&\frac{1}{N_1}\sum_{i=1}^{N_1}z_i^{(1)}z_i^{(2)}\nn,\\
 q_{12,N_2}&=&\frac{1}{N_2}\sum_{i=N_1+1}^{N_1}z_i^{(1)}z_i^{(2)}.
 \eea
 Since $q_{12,N}$ is a convex linear combination of $q_{12,N_1}$ and $q_{12,N_2}$,
 \begin{equation}
 q_{12,N}=\frac{N_1}{N}q_{12,N_1}+\frac{N_2}{N}q_{12,N_2},
 \end{equation}
 and due to the convexity of the function $x\rightarrow x^2$, we have the inequality
 \be\label{e7}
 \left\langle q_{12,N}^2-\frac{N_1}{N} q_{12,N_1}^2-\frac{N_2}{N} q_{12,N_2}^2\right\rangle \leq 0.
 \ee
 A combination of the informations in ($\ref{e6}$) and ($\ref{e7}$) allows us to state the following result.
 \begin{lemma}
 The interpolating function is increasing in $t$ i.e. $\frac{d}{dt}\varphi (t)\geq 0.$
 \end{lemma}
 \nin By integrating in $t$ we get
 \be
 \varphi(1)=\varphi (0)+\int_0^1 \frac{d}{dt}\varphi(t) dt\geq \varphi(0)
 \ee
 and recalling the boundary conditions ($\ref{e9}$) we obtain the main result
 \begin{teorema}
 The following superadditivity property holds
 \be
 NA_N(\beta,h)\geq N_1A_{N_1}(\beta,h)+N_2A_{N_2}(\beta,h).
 \ee
 \end{teorema}
 \nin The superadditivity property gives an immediate control of the thermodynamic limit \cite{Ru69}, and we can state the next
 \begin{teorema}
 The thermodynamic limit for $A_N(\beta,h)$ exists and equals its $\sup$ i.e.
 \begin{equation}
 \lim_{N\rightarrow\infty}A_N(\beta,h)=A(\beta,h)=\sup_NA_N(\beta,h).
 \end{equation}
 \end{teorema}

\section{High Temperature behavior}
 We start to analyze our model characterizing the high temperature regime at zero external field. First we define the annealed free energy of the model
 \be\label{he1}
 -\beta f_N^A(\beta,\lambda)=A_N^A(\beta,\lambda)=\frac{1}{N}\log \E Z_N(\beta,\lambda,J),
 \ee
 that can be easily computed as in the following
 \begin{proposizione}
 For $\lambda<1$ the annealed free energy of the model in the thermodynamic limit is well defined and coincides with
 \be\label{p1}
 -\beta f^A(\beta,\lambda)=\lim_{N\rightarrow \infty}A^A_N(\beta,\lambda)=-\frac{1}{2}\log(1-\lambda).
 \ee
 \end{proposizione}
\begin{proof}
\bea\label{z1}
 \mathbb{E}_JZ_N&=&\mathbb{E}_J\mathbb{E}_{z}\exp\left(\frac{\beta}{\sqrt{2N}}\sum_{i,j=1}^NJ_{ij}z_iz_j-\frac{\beta^2}{4N}(\sum_{i=1}^Nz_i^2)^2
 +\frac{\lambda}{2}\sum_{i=1}^Nz_i^2\right)\nn\\
 &=&\mathbb{E}_{z}\exp\left(-\frac{\beta^2}{4N}(\sum_{i=1}^Nz_i^2)^2
 +\frac{\lambda}{2}\sum_{i=1}^Nz_i^2\right)\mathbb{E}_{J}\exp\left(\frac{\beta}{\sqrt{2N}}\sum_{i,j=1}^NJ_{ij}z_iz_j\right)\nn\\
 &=&\mathbb{E}_{z}\exp\left(-\frac{\beta^2}{4N}(\sum_{i=1}^Nz_i^2)^2
 +\frac{\lambda}{2}\sum_{i=1}^Nz_i^2\right)\exp\left(\frac{\beta^2}{4N}\sum_{i,j=1}^Nz_i^2z_j^2\right)\nn\\
 &=&\mathbb{E}_{z}\exp\left(\frac{\lambda}{2}\sum_{i=1}^Nz_i^2\right)\nn\\
 &=&(1-\l)^{-\frac{N}{2}}
 \eea
 Thus ($\ref{p1}$) follows from ($\ref{he1}$) and the proposition is proven.
\end{proof}
\nin We define the high temperature regime as the region in the $(\beta,\lambda)$ plane where the quenched free energy is equal to the annealed one. We already know that the annealed approximation is an upper bound for the pressure, infact a simple application of the Jensen inequality shows that
\be\label{annleq}
\frac 1 N\E\log Z_N(\b,\l; J)\leq \frac 1 N\log \E Z_N(\b,\l,J)=A^A(\b,\l).
\ee
On the other side we have that
\be\label{anngeq}
\frac 1 N\E\log Z_N(\b,\l; J)\geq \frac 1 N\E\log Z'_N(\b,\l; J),
\ee
where $Z'_N(\b,\l;J)$ is an auxiliary partition function in which diagonal terms of the spin-spin interaction are neglected, i.e.
\bea
Z'_N(\b,\l;J)&=&\E_z\exp\left(-\frac{1}{\sqrt{2N}}\sum_{i\neq j}^NJ_{ij}z_iz_j-\frac{\beta^2}{4N}(\sum_{i=1}^N z_i^2)^2+\frac{\lambda}{2}\sum_{i=1}^Nz_i^2\right)\nn\\
&=&\E_z\exp\left(-\frac{1}{\sqrt{N}}\sum_{i<j}^NJ_{ij}z_iz_j-\frac{\beta^2}{4N}(\sum_{i=1}^N z_i^2)^2+\frac{\lambda}{2}\sum_{i=1}^Nz_i^2\right),\nn
\eea 
where we have noted that $\frac 1 {\sqrt{2}}(J_{ij}+J_{ji})$ is a centered gaussian random variable $\mathcal{N}(0,1)$ that we have simply denoted by $J_{ij}$. Inequality (\ref{anngeq}) follows by an other application of the Jensen inequality on the $J_{ii}$ noises:
\bea
\E\log Z_N(\b,\l; J)&=&\E_{J_{ij}}\E_{J_{ii}}\log Z_N(\b,\l; J_{ij},J_{ii})\nn\\
&\geq& \E_{J_{ij}}\log Z_N(\b,\l; J_{ij},\E_{J_{ii}}[J_{ii}])\nn\\
&=&\E_{J_{ij}}\log Z_N(\b,\l; J_{ij},0)=\E\log Z'_N(\b,\l; J),\nn
\eea
Note that the auxiliary partition function $Z'_N$ gives the same annealed approximation of $Z_N$; infact we have the following
\begin{proposizione}
For $\l<1$,
\be
\lim_{N\rightarrow\infty}\frac 1 N \log \E_J Z'_N(\b,\l;J)= -\frac{1}{2}\log(1-\lambda)= A^A(\b,\l)
\ee
\end{proposizione}
\begin{proof}
\bea\label{z'2}
 \mathbb{E}_JZ'_N&=&\mathbb{E}_J\mathbb{E}_{z}\exp\left(\frac{\beta}{\sqrt{N}}\sum_{i<j}^NJ_{ij}z_iz_j-\frac{\beta^2}{4N}(\sum_{i=1}^Nz_i^2)^2
 +\frac{\lambda}{2}\sum_{i=1}^Nz_i^2\right)\nn\\
 &=&\mathbb{E}_{z}\exp\left(-\frac{\beta^2}{4N}(\sum_{i=1}^Nz_i^2)^2
 +\frac{\lambda}{2}\sum_{i=1}^Nz_i^2\right)\mathbb{E}_{J}\exp\left(\frac{\beta}{\sqrt{N}}\sum_{i<j}^NJ_{ij}z_iz_j\right)\nn\\
 &=&\mathbb{E}_{z}\exp\left(-\frac{\beta^2}{4N}(\sum_{i=1}^Nz_i^2)^2
 +\frac{\lambda}{2}\sum_{i=1}^Nz_i^2\right)\exp\left(\frac{\beta^2}{2N}\sum_{i<j}^Nz_i^2z_j^2\right)\nn\\
 &=&\mathbb{E}_{z}\exp\left(\frac{\lambda}{2}\sum_{i=1}^Nz_i^2-\frac{\beta^2}{4N}\sum_{i=1}^Nz_i^4\right)\nn\\
 &=&(1-\lambda)^{-\frac{N}{2}}\left(\int \frac{dz}{\sqrt{2\pi}}e^{-\frac{1}{2}z^2-\frac{\beta^2}{4N(1-\lambda)^2}z^4}\right)^N
 \eea
Now, putting $\b_\l=\frac{\b}{1-\l}$, we notice that the function in the integral
 \be\label{eq:besselK}
 \int\frac{dz}{\sqrt{2\pi}}e^{-\frac{1}{2}z^2-\frac{\beta_\l^2}{4N}z^4}
 \ee
approaches to 1 uniformly for $0\leq\lambda<1$ when $N$ grows to infinity, and so the integral, that completes the proof.
\end{proof}

\nin Now, we can control  the high temperature region of $Z'_N$ studying the fluctuations of the random variable $Z'_N/\E Z'_N$ in according to the Borel-Cantelli lemma approach \cite{HT}\cite{T}. The following lemma holds:
 \begin{lemma}\label{lann1}
For $\beta_{\lambda}=\frac{\beta}{1-\lambda}\leq 1$ we have
 \be
 \limsup_{N\rightarrow\infty}\frac{\mathbb{E}_J(Z_N^{'2})}{\mathbb{E}^2_J(Z'_N)}\leq \frac{1}{\sqrt{1-\beta_{\lambda}^2}}.
 \ee
 \end{lemma}
 \nin Before proving lemma $\ref{lann1}$ we note that it is a sufficient condition to state the following
 \begin{lemma}\label{ann}
 In the region of the $(\beta,\lambda)$ plane defined by $\beta_\lambda<1$, i.e. $\beta<1-\lambda$.
 \be
 \lim_{N\rightarrow\infty}\frac 1 N \E\log Z_N'(\b,\l;J)=\lim_{N\rightarrow\infty}\frac 1 N\log \E Z'_N(\b,\l;J)=A^A(\b,\l)
 \ee
 \end{lemma}
 \nin Thanks to inequalities (\ref{annleq}) and (\ref{anngeq}), we have proven the following main
 \begin{teorema}
 The quenched free energy of the Gaussian spin glass model at zero external field does coincide with the annealed one
 \be
  -\b f(\b,\l)=\lim_{N\rightarrow\infty}\frac 1 N \E\log Z_N(\b,\l;J)= -\frac{1}{2}\log(1-\lambda)
 \ee
 in the region of the $(\beta,\lambda)$ plane defined by $\beta<1-\lambda$.
 \end{teorema}
 \nin Now we attack Lemma \ref{lann1}. 
 \begin{proof}
 At first we evaluate $\E(Z_N^{'2})$. By a straightforward calculation we have
 \bea\label{z2}
 \E(Z_N^{'2})\leq\E_{1,2}\exp\left(-\frac{\beta^2}{4N}\sum_{i=1}^N{z^{(1)}_i}^4+{z^{(2)}_i}^4+
 \frac{\lambda}{2}\sum_{i=1}^N{z^{(1)}_i}^2+{z^{(2)}_i}^2
 +\frac{\beta^2N}{2}q_{12}^2\right),\nn
 \eea
 where we neglected a term $e^{-\frac{\beta^2}{2N}\sum_iz^{(1)^2}_iz^{(2)^2}_i}<1$. We can linearize the overlap term introducing an auxiliary
 $\mathcal{N}(0,1)$ gaussian random variable $g$ so $e^{\frac{\beta^2N}{2}q_{12}^2}=\E_ge^{\beta\sqrt{N}q_{12}g}$. In this way the average
 factorizes on $i$ so, bearing in mind $(\ref{z'2})$ and the overlap's definition, we can write
 \bea
 \frac{\E(Z_N^{'2})}{\E^2(Z'_N)}&\leq&\E_g\left(\frac{\E_{x,y}e^{-\frac{\beta^2}{4N}(x^4+y^4)+\frac{\lambda}{2}(x^2+y^2)+
 \frac{\beta}{\sqrt{N}}xyg}}{\mathbb{E}_{x,y}e^{-\frac{\beta^2}{4N}(x^4+y^4)+\frac{\lambda}{2}(x^2+y^2)}}\right)^N\nn\\
 &=&\E_g\left(\frac{\E_{x,y}e^{-\frac{\beta_{\lambda}^2}{4N}(x^4+y^4)+\frac{\beta_{\lambda}}{\sqrt{N}}xyg}}
 {\mathbb{E}_{x,y}e^{-\frac{\beta_{\lambda}^2}{4N}(x^4+y^4)}}\right)^N,
 \eea
where we used the simpler notation $x=z^{(1)}_1$, $y=z^{(2)}_1$. Now it is sufficient to note that
 \bea
 \mathbb{E}_{x,y}e^{-\frac{\beta_{\lambda}^2}{4N}(x^4+y^4)+\frac{\beta_{\lambda}}{\sqrt{N}}xyg}&=&
 \E_x e^{-\frac{\beta^2_{\lambda}}{4N}x^4+\frac{\beta_{\lambda}^2}{2N}g^2x^2}\int\frac{dy}{\sqrt{2\pi}}
 e^{-\frac{1}{2}(y-\frac{\beta_{\lambda}}{\sqrt{N}}xg)^2-\frac{\beta^2_{\lambda}}{4N}y^4}\nn\\
 &=&\E_x e^{-\frac{\beta^2_{\lambda}}{4N}x^4+\frac{\beta_{\lambda}^2}{2N}g^2x^2}\E_y e^{-\frac{\beta^2_{\lambda}}{4N}(y+\frac{\beta_{\lambda}}{\sqrt{N}}xg)^4}\nn\\
 &\leq&\E_x e^{-\frac{\beta^2_{\lambda}}{4N}x^4+\frac{\beta_{\lambda}^2}{2N}g^2x^2}\E_y e^{-\frac{\beta^2_{\lambda}}{4N}y^4}\nn,
 \eea
 since the function $\E_y[e^{-(y+a)^4}]$ is concave in $a$, and it exhibits a unique maximum in $a=0$, as it can be easily verified. Hence we finally get
 \be
 \frac{\E(Z_N^{'2})}{\E^2(Z'_N)}\leq\E_g\left[ \frac{\int \frac{dx}{\sqrt{2\pi}}e^{-\frac{x^2}{2}} e^{\frac{\b^2_\l x^2g^2}{2N}} e^{-\frac{\b^2_\l x^4}{4N} }}{ \int \frac{dx}{\sqrt{2\pi}}e^{-\frac{x^2}{2}} e^{-\frac{\b^2_\l x^4}{4N} }} \right]^N,
 \ee
 and we can split this integral as a sum over the complementary regions $\{g^2<N/\b^2_\l\}$ and $\{g^2\geq N/\b^2_\l\}$. For the first one, we see that 
 \bea
& & \lim_{N\rightarrow\infty}\int_{g^2<N/\b^2_\l}\frac{dg}{\sqrt{2\pi}}e^{-\frac{g^2}{2}}\left[ \frac{\int \frac{dx}{\sqrt{2\pi}}e^{-\frac{x^2}{2}} e^{\frac{\b^2_\l x^2g^2}{2N}} e^{-\frac{\b^2_\l x^4}{4N} }}{ \int \frac{dx}{\sqrt{2\pi}}e^{-\frac{x^2}{2}} e^{-\frac{\b^2_\l x^4}{4N} }} \right]^N\nn\\
&=&\lim_{N\rightarrow\infty}\int_{g^2<N/\b^2_\l}\frac{dg}{\sqrt{2\pi}}e^{-\frac{g^2}{2}}(1-\frac{\beta_{\lambda}^2}{N}g^2)^{-\frac{N}{2}}\left(\frac{\int \frac{dx}{2\pi}e^{-\frac{1}{2}x^2-\frac{\beta^2_{\lambda}}{4N(1-\frac{\beta_{\lambda}^2}{2N}g^2)^2}x^4}}{\int\frac{dx}{2\pi}e^{-\frac{1}{2}x^2-\frac{\beta^2_{\lambda}}{4N}x^4}}\right)^N\nn\\
&\leq&\lim_{N\rightarrow\infty} \int_{g^2<N/\b^2_\l}\frac{dg}{\sqrt{2\pi}}e^{-\frac{g^2}{2}}\left( 1-\frac{\b^2_\l g^2}{N} \right)^{-\frac{N}{2}}\nn\\
&\leq& \frac{1}{\sqrt{1-\b^2_\l}}.
 \eea
For the second one, we cannot perform the change of variable for the variance of the gaussian, and we have to estimate the integrals. As we have seen, the integral in the denominator is given by (\ref{eq:besselK}). Thus we can rewrite the second term as
$$
\int_{g^2\geq N/\b^2_\l}\frac{dg}{\sqrt{2\pi}}e^{-\frac{g^2}{4}}\frac{e^{-g^2/4}\left(\int \frac{dx}{\sqrt{2\pi}}e^{-\frac{x^2}{2}} e^{\frac{\b^2_\l x^2g^2}{2N}} e^{-\frac{\b^2_\l x^4}{4N} }\right)^N}{ \left(\int \frac{dx}{\sqrt{2\pi}}e^{-\frac{x^2}{2}} e^{-\frac{\b^2_\l x^4}{4N} }\right)^N},
$$
and notice that the function at the numerator in the integral is concave in $g$, and it assumes the unique maximum point at $g=0$, where it attains the same value of the denominator. Therefore we get the super-exponential decay of the gaussian tails:
\bea
& &\int_{g^2\geq N/\b^2_\l}\frac{dg}{\sqrt{2\pi}}e^{-\frac{g^2}{4}}\frac{e^{-g^2/4}F_N^N(g,\b_\l)}{ \left(\int \frac{dx}{\sqrt{2\pi}}e^{-\frac{x^2}{2}} e^{-\frac{\b^2_\l x^4}{4N} }\right)^N}\nn\\
&\leq&\sqrt{2}P(g^2\geq N/\b^2_\l)\simeq C\sqrt{\frac{\b^2_\l}{N}}e^{-\frac{N}{4\b^2_\l}},
\eea
for a certain constant $C$, and the lemma is proven.
\end{proof}

\section{Sum Rules for the Free Energy}
 In this section we introduce the replica symmetric approximation for the free energy density. In particular, we obtain it as an upper bound for $-\beta f$
 togheter with the error, in the form of a sum rule. For this purpose, we apply a well known interpolation scheme \cite{sum-rules}\cite{io1}\cite{Bip} \cite{NN} to compare the
 original two-body interaction with a one-body interaction system. Concretely, we define, for $t\in[0,1]$ and $\bar{q}\geq 0$\footnote{despite it will be transparent at the end of the section, it may result helpful to bear in mind that $\bar{q}$ will act as the replica symmetric approximation of the overlap.}, the interpolating partition function
 \bea\label{rs1}
 Z_N(t,J,J')&=&\E_z\exp\left(\beta\sqrt{\frac{t}{2N}}\sum_{i,j=1}^N J_{ij}z_iz_j -t\frac{\beta^2}{4N}(\sum_{i=1}^N z_i^2)^2\right.\nn\\
 &&\left.+\beta\sqrt{1-t}\sqrt{\bar{q}}\sum_{i=1}^NJ'_iz_i +(1-t)\frac{c}{2}\sum_{i=1}^Nz_i^2\right)\nn\\
 &&\exp\left(\frac{\lambda}{2}\sum_{i=1}^Nz_i^2\right),
 \eea
where the external noise $J'_i$ are i.i.d Gaussian random variable $\mathcal{N}(0,1)$ and they are also independent from all $J_{ij}$. Here $c$ is an additional lagrangian multiplier to be fixed later. Now we introduce the interpolating function
 \be\label{interp1}
 \varphi_N(t)=\frac{1}{N}\E\log Z_N(t,J,J'),
 \ee
 where we encode in $\E$ the averages respect to both $J$ and $J'$. At $t=1$ the interpolating function (\ref{interp1})  recovers the original system, while at $t=0$ it accounts for a simpler factorized one-body
 model and we can easily get
 \bea
 \varphi_N(0)&=&\frac{1}{N}\E\log\prod_{i=1}^N\E_{z_i}\exp\left(\beta\sqrt{\bar{q}}J'_iz_i + \frac{(c+\l)}{2}z_i^2\right)\nn\\
 &=&\E_{J'}\log\E_g\exp\left(\beta\sqrt{\bar{q}}J'g +\frac{(c+\l)}{2}g^2\right)\nn\\
 &=& \E_{J'}\log(1-\l-c)^{-\frac{1}{2}}\E_g\exp\left(\frac{\beta\sqrt{\bar{q}}}{(1-\l-c)^{\frac{1}{2}}}J'g \right)\nn\\
 &=& \log(\sigma)+\frac{1}{2}\E_{J'}\beta^2\bar{q}\sigma^2J'^2=\log(\sigma)+\frac{1}{2}\beta^2\bar{q}\sigma^2,
 \eea
with $\sigma=(1-\l-c)^{-\frac{1}{2}}$ and where $g$ as usual is a $\mathcal{N}(0,1)$ random variable. Therefore $\varphi_N(t)$ fulfills the following boundary conditions:
 \bea
 \varphi_N(1)&=&A_N(\beta,\lambda)\nn\\
 \varphi_N(0)&=&\log(\sigma)+\frac{1}{2}\beta^2\bar{q}\sigma^2.
 \eea
 Now we have to evaluate the $t-$derivative of $\varphi_N(t)$ in order to obtain the sum rule
 \be
 \varphi_N(1)=\varphi_N(0)+\int_0^1 dt\frac{d}{dt} \varphi_N(t).
 \ee
Using the notation $\meanv{.}_t=\E\Omega_t(.)$, where $\Omega_t(.)$ is the replicated Boltzmann state encoded in the partition function $(\ref{rs1})$, we can write
 \bea\label{rs2}
 \frac{d}{dt}\varphi_N(t)&=&\frac{1}{N}\frac{\beta}{2\sqrt{2tN}}\sum_{i,j=1}^N\E J_{ij}\omega_t(z_iz_j)-\frac{\beta^2}{4N^2}\meanv{(\sum_{i=1}^Nz_i^2)^2}_t\nn\\
 &-&\frac{1}{N}\frac{\beta\sqrt{\bar{q}}}{2\sqrt{1-t}}\sum_{i=1}^N\E J'_i\omega_t(z_i)-\frac{c}{2N}\sum_{i=1}^N\meanv{z_i^2}_t.
 \eea
 A standard integration by parts over the external noise, as in $(\ref{ew})$, shows that
 \bea\label{rs3}
 \frac{1}{N}\frac{\beta}{2\sqrt{2tN}}\sum_{i,j=1}^N\E J_{ij}\omega_t(z_iz_j)&=&
 \frac{\beta^2}{4N^2}\meanv{(\sum_{i=1}^Nz_i^2)^2}_t-\frac{\beta^2}{4}\meanv{q_{12}^2}_t\nn\\
 \frac{1}{N}\frac{\beta\sqrt{\bar{q}}}{2\sqrt{1-t}}\sum_{i=1}^N\E J'_i\omega_t(z_i)&=&-\frac{\beta^2}{2}\bar{q}\meanv{q_{12}}_t+
 \frac{\beta^2\bar{q}}{2N}\sum_{i=1}^N\meanv{z_i^2}_t.
 \eea
 Inserting $(\ref{rs3})$ into $(\ref{rs2})$, we get
 \be
 \frac{d}{dt}\varphi_N(t)=-\frac{\beta^2}{4}\meanv{q_{12}^2}_t+\frac{\beta^2}{2}\bar{q}\meanv{q_{12}}_t
 +\frac{1}{2N}\sum_{i=1}^N(-\beta^2\bar{q}-c)\meanv{z_i^2}_t,
 \ee
 hence, adding and subtracting a term $\frac{\beta^2}{4}\bar{q}^2$ and with the choice $c=-\beta^2\bar{q}$, we finally obtain
 \be
 \frac{d}{dt}\varphi_N(t)=\frac{\beta^2}{4}\bar{q}^2-\frac{\beta^2}{4}\meanv{(q_{12}-\bar{q})^2}_t.
 \ee
 We have just proved the following

 \begin{teorema}\label{rslem}
For every $\bar q\in\mathcal{D}_{\beta,\lambda}(\bar{q})\equiv
 \left\{\bar{q}\in\R^{+} : 1-\lambda+\beta^2\bar{q}>0\right\}$ is defined
 \be\label{rs5}
 \tilde{A}(\beta,\lambda,\bar{q})=\log(\sigma)+\frac{1}{2}\beta^2\bar{q}\sigma^2+\frac{\beta^2}{4}\bar{q}^2
 \ee
with $\sigma=(1-\lambda+\beta^2\bar{q})^{-\frac{1}{2}}$. Then, $\forall N$ and $\forall\bar{q}\in\mathcal{D}_{\beta,\lambda}(\bar{q}) $,
 the quenched free energy of the mean field gaussian spin glass model defined in $(\ref{e2})$ fulfills the sum rule
 \be\label{rs'6}
 A_N(\beta,\lambda)=-\beta f_N(\beta,\lambda)=\tilde{A}(\beta,\lambda,\bar{q})-\frac{\beta^2}{4}\int_0^1dt\meanv{(q_{12}-\bar{q})^2}_t.
 \ee
 Moreover, $\forall\bar{q}\in\mathcal{D}_{\beta,\lambda}(\bar{q}) $, $\tilde{A}(\beta,\lambda,\bar{q})$ is an upper bound for $A_N(\beta,\lambda)$
 uniformly in $N$, i.e.
 \be\label{rs'5}
 A_N(\beta,\lambda)=-\beta f_N(\beta,\lambda)\leq\tilde{A}(\beta,\lambda,\bar{q}).
 \ee
 \end{teorema}
 \nin Since the bound $(\ref{rs'5})$ is uniform in $N$, then it is true also in the thermodynamic limit. The error term in $(\ref{rs'6})$ reduces to the overlap's fluctuations around $\bar{q}$. We can minimize this error, or equivalently
 optimize the estimate in $(\ref{rs'5})$, by taking the value of $\bar{q}$ that minimize $\tilde{A}(\beta,\lambda,\bar{q})$. For this purpose we state the following

 \begin{proposizione}\label{rsprop}
 We have
 $$
\frac{\partial}{\partial\bar{q}}\tilde{A}(\beta,\lambda,\bar{q})=0\quad\Leftrightarrow\quad\bar{q}=0\quad\mbox{or}\quad\bar{q}=\frac{\beta-(1-\l)}{\beta^2}.
 $$
\nin The solution $\bar{q}=0$ is a minimum for $\b_\l<1$, i.e. $\beta\leq1-\lambda)$. Conversely, for $\beta>1-\lambda$ the minimum is $\bar{q}=\frac{\beta-(1-\l)}{\beta^2}$.
 \end{proposizione}
 \begin{proof}
 Since $\partial_{\bar q}\sigma=-\frac{\b^2}{2}\s^3$ we have that
 \bea
 \frac{\partial}{\partial\bar{q}}\tilde{A}(\beta,\lambda,\bar{q})
 &=&-\frac{\b^2}{2}\s^2+\frac{\b^2}{2}\s^2+\frac{\b^4}{2}\bar q\s^4+\frac{\b^2}{2}\bar q\nn\\
 &=&\frac{\b^2}{2}\bar q\left(1-\b^2\s^4\right)\nn
 \eea
 with two roots $\bar q =0$ and $\b\s^2=1$, i.e. $\bar q=\frac{\beta-(1-\l)}{\beta^2}$. If we study $\tilde{A}(\beta,\lambda,\bar{q})$ as a function of $\bar{q}^2$ we see that
 \bea
 \frac{\partial}{\partial\bar{q}^2}\tilde{A}(\beta,\lambda,\bar{q})&=&\frac{1}{2\bar q}\frac{\partial}{\partial\bar{q}}\tilde{A}(\beta,\lambda,\bar{q})=\frac{\beta^2}{4}\left(1-\frac{\b^2}{(1-\l+\b^2\bar q)^2}\right)\nn.
 \eea
 Since $\frac{\partial}{\partial\bar{q}^2}\tilde{A}(\beta,\lambda,\bar{q})$ is increasing, $\tilde A$ is a convex function of $\bar{q}^2$ and at $\bar q=0$ we have that
 $$
 \frac{\partial}{\partial\bar{q}^2}\tilde{A}(\beta,\lambda,\bar{q}^2)|_{\bar q =0}=\frac{\b^2}{4}\left(1-\frac{\b^2}{(1-\l)^2}\right)
 =\frac{\b^2}{4}\left(1-\b_\l^2\right).
 $$
Due to the convexity of $\tilde A(\bar{q}^2)$, the minimum is achieved at $\bar q =0$ for $\b_\l<1$ and at $\bar q>0$ for $\b_\l>1$.
 \end{proof}
 \nin By combining the information of Theorem $\ref{rslem}$ and Proposition $\ref{rsprop}$ we have the proof of the following main result.
 \begin{teorema}\label{RSteor}
 The replica symmetric approximation for the free energy is well defined by the following variational principle:
 \be
 A^{RS}(\beta,\lambda)=\inf_{\bar{q}\in\mathcal{D}_{\beta,\lambda}(\bar{q})}\tilde{A}(\beta,\lambda,\bar{q}),
 \ee
 where
 \be
 \tilde{A}(\beta,\lambda,\bar{q})=\log(\sigma)+\frac{1}{2}\beta^2\bar{q}\sigma^2+\frac{\beta^2}{4}\bar{q}^2,
 \ee
 with $\sigma(\beta,\lambda,\bar{q})$ defined in $(\ref{rs5})$.  The minimum is achieved at $\bar{q}=0$ for $\beta\leq 1-\lambda$ and at
 $\bar{q}=\frac{\beta-(1-\l)}{\beta^2}$ otherwise. Moreover the replica symmetric approximation is an upper bound for $A(\beta,\lambda)$, infact,
 uniformly in $N$,
 \be\label{rs6}
 A_N(\beta,\lambda)=-\beta f_N(\beta,\lambda)\leq A^{RS}(\beta,\lambda).
 \ee
 \end{teorema}
 \nin For $\beta_\l<1$ the replica symmetric free energy reduces to the annealed one, that, accordingly with Theorem $\ref{ann}$, coincides with the thermodynamic limit of the true free energy in such a region. Note that $\bar{q}=\frac{\beta-(1-\l)}{\beta^2}$ is also the optimal value for $\l>1$, infact $1-\l+\b^2\bar q=\b>0$ such that $\bar q \in \mathcal{D}_{\beta,\lambda}(\bar{q})$ and the RS approximation is well defined. In this case we see that $\bar q \rightarrow \infty$ when $\b\rightarrow 0$.

\section{Fluctuation Theory for the Order Parameter}

This section is dedicated to the study of the fluctuations of the (rescaled and centered) order parameter.

\nin The general idea behind is that critical phenomena arises in presence of a divergence of the fluctuation  of the order parameter of the model. As critical phenomena are interesting by themselves, this analysis deserve major depth. In particular, we want to bound the annealed region, namely where  $\bar{q}=\lim_{N \to \infty}\meanv{q^2_{12,N}}=0$, checking that the rescaled fluctuations of $\meanv{q_{12,N}}$ diverge on the same critical line where the fluctuation of the annealed free energy are singular. 

\nin To this task we introduce and define the rescaled and centered  overlap
\be
\xi_{12,N}=\sqrt{N}(q_{12,N}-\bar{q}),
\ee
which, in the thermodynamic limit, converges to a Gaussian random variable, whose variance spreads up to infinity as far as the system approaches the critical line in the $(\beta,\lambda)$ plane.
\newline
Once again the strategy we outline is the one developed for the Sherrington-Kirkpatrick model \cite{sum-rules}\cite{GTQ}\cite{GTF}. It is still based on the evaluation of overlap correlations at $t=0$ with respect to the \textit{Boltzmannfaktor} defined in (\ref{rs1}): We can in fact evaluate the thermodynamical observables at $t=0$ due to the lacking of correlation and then propagate the solution up to $t=1$.
\newline
To this task we introduce the following proposition
\begin{proposizione}\label{stream}
For every smooth function $F_s$ of the overlaps $\left\{q_{ab}\right\}_{1\leq a<b\leq s}$ among $s$ replicas,
\be
\frac{d}{dt}\meanv{F_s}_t=\frac{\beta^2}{2}\left(\sum_{1\leq a<b\leq s}\meanv{F_s \xi^2_{ab}}_t-s\sum_{a=1}^s\meanv{F_s \xi^2_{as+1}}_t +
\frac{s(s+1)}{2}\meanv{F_s \xi^2_{s+1s+2}}_t \right).
\nn\ee
\end{proposizione}
\nin The proof is here omitted, since it can be easily obtained by long but direct calculation.
\nin  We are interested in considering $F_s=\xi^2_{12}$, such that
\be
\frac{d}{dt}\meanv{\xi_{12}^2}_t=\frac{\beta^2}{2}\left(\meanv{\xi^4_{12}}_t-4\meanv{\xi^2_{12}\xi_{13}^2}_t+3\meanv{\xi^2_{12}\xi^2_{34}}_t\right).
\ee
To understand how $\meanv{\xi_{12}^2}_t$ behaves we need to tackle even the other two correlation functions $\meanv{\xi_{12}\xi_{13}}_t$, and $\meanv{\xi_{12}\xi_{34}}_t$.
For the sake of simplicity, let us consider $t$ as a time and put
\be
A(t)=\meanv{\xi_{12}^2}_t,\ \ \ B(t)=\meanv{\xi_{12}\xi_{13}}_t,\ \ \ C(t)=\meanv{\xi_{12}\xi_{34}}_t.
\ee
Under the Gaussian \textit{Ansatz} for the high temperature behavior of $\xi_{ab}$ we can apply Wick theorem and, by using Proposition (\ref{stream}), we can construct the following dynamical system for $A(t)$, $B(t)$ and $C(t)$:
\bea
\dot{A}&=&\beta^2(A^2-4B^2+3C^2),\nn\\
\dot{B}&=&\beta^2(2AB-6BC+6C^2-2B^2),\nn\\
\dot{C}&=&\beta^2(\frac{1}{2}AC+4B^2-16BC+10C^2),\nn
\eea
which can be straightforwardly solved with the initial data
\bea
A(0)&=&\E\omega^2(z_i^2)-\bar{q}^2,\nn\\
B(0)&=&\E\omega(z_i^2)\omega^2(z_i)-\bar{q}^2,\nn\\
C(0)&=&\E\omega(z_i)^4-\bar{q}^2,\nn
\eea
where
\bea
\omega(z_i)&=&\frac{\E_z z e^{\beta\sqrt{\bar{q}}Jz+\frac{(c+\l)}{2}z^2}}{\E_z e^{\beta\sqrt{\bar{q}}Jz+\frac{(c+\l)}{2}z^2}}=\frac{1}{\beta\sqrt{\bar{q}}}\partial_J \log \E_z \exp(\beta\sqrt{\bar{q}}Jz+\frac{(c+\l)}{2}z^2)\nn\\
&=&\frac{1}{\beta\sqrt{\bar{q}}}\partial_J(\log(\sigma)+\frac{1}{2}\beta^2\sigma^2\bar{q}J^2)=\beta\sqrt{\bar{q}}\sigma^2J,
\eea
\bea
\omega(z_i^2)&=&\frac{\E_z z^2 e^{\beta\sqrt{\bar{q}}Jz+\frac{(c+\l)}{2}z^2}}{\E_z e^{\beta\sqrt{\bar{q}}Jz+\frac{(c+\l)}{2}z^2}}
=\frac{1}{\beta^2\bar{q}}\frac{\partial^2_J\E_z e^{\beta\sqrt{\bar{q}}Jz+\frac{(c+\l)}{2}z^2}}{\E_z e^{\beta\sqrt{\bar{q}}Jz+\frac{(c+\l)}{2}z^2}}\nn\\
&=&\frac{1}{\beta^2\bar{q}}\left(\partial^2_J(\log\E_z e^{\beta\sqrt{\bar{q}}Jz+\frac{(c+\l)}{2}z^2})
+(\partial_J\log\E_z e^{\beta\sqrt{\bar{q}}Jz+\frac{(c+\l)}{2}z^2})^2\right)\nn\\
&=&\sigma^2+\beta^2\bar{q}\sigma^4J^2.
\eea
As we are interested in finding criticality, the simplest procedure is approaching the critical line from the annealed regime, where $\bar{q}=0$. This further simplifies the initial conditions as
\bea
A(0)=\sigma^4,\ \ \ B(0)=C(0)=0.
\eea
So the solution of the dynamical system is trivial as $B(t)$ and $C(t)$ are identically zero, while $A(t)$  satisfies
\be
\dot{A}=\beta^2A^2,\ \ \ A(0)=\sigma^4,
\ee
whose solution is
\be
A(t)=\frac{1}{\sigma^{-4}-\beta^2t}.
\ee
Remembering that $\sigma=(1-\lambda+\beta^2\bar(q))^{-\frac{1}{2}}$, propagating up to $t=1$ we finally obtain
\be
\meanv{\xi_{12}^2}=A(1)=\frac{1}{(1-\lambda)^2-\beta^2},
\ee
that diverges when $\beta_{\lambda}=1$, i.e. $\beta= 1-\lambda$, in complete agreement with Theorem 3.

\section{Broken Replica Symmetry Bound}
In this section we go beyond the replica symmetric approximation and we show a different bound for the free energy density, that should in principle improve the previous one. First of
 all we introduce the convex space $\chi$ of functional order parameters $x$, as nondecreasing functions of the auxiliary variable $q$ in the $[0,1]$ interval,
 i.e.
 \be
  \chi \ni x : [0,Q] \ni q \rightarrow x(q)\in [0,1],
 \ee
 and we have to think $x(q)$ as a possible distribution function for the overlap. We will consider the case of piecewise constant functional order parameters,
 characterized by an integer $K$ and two sequence of numbers, $q_0, q_1,\ldots, q_K $ and $m_1,\ldots, m_K$, satisfying
 \be\label{ordpar}
 0=q_0\leq q_1\ldots\leq q_K=Q\ \ 0\leq m_1\ldots \leq m_K\leq 1,
 \ee
 such that $x(q)=m_i$ for $q\in[q_{i-1},q_i]$. It is useful to define also $m_0=0$ and $m_{K+1}=1$. The replica symmetric case correspond to $K=2$, $q_1=\bar{q}$,
 $m_1=0$ and $m_2=1$, where overlap selfaverages around $\bar{q}$; the case $K=3$, with two possible value ($q_1$ and $q_2$) for the overlap, is the first level
 of replica symmetry breaking, and so on. Now, following the interpolation scheme in \cite{RSB}, we consider a generic piecewise constant $x(q)$ and we introduce
 the interpolating partition function
 \bea \label{RSB1}
 \tilde Z_N(t;x(q))=\E_z&&\exp\left(\beta\sqrt{\frac{t}{2N}}\sum_{i,j=1}^NJ_{ij}z_iz_j-t\frac{\beta^2}{4N}(\sum_{i=1}^Nz_i^2)^2\right.\nn\\
 &&\left.+\beta\sqrt{1-t}\sum_{a=1}^K\sqrt{q_a-q_{a-1}}\sum_{i=1}^NJ^a_iz_i+(1-t)\frac{C}{2}\sum_{i=1}^Nz_i^2\right)\nn\\
 &&\exp\left(\beta h\sum_{i=1}^Nz_i+\frac{\lambda}{2}\sum_{i=1}^Nz_i^2\right),
 \eea
 where $t\in[0,1]$. Here we have introducted additional independent gaussian random variable $J_i^a \in\mathcal{N}(0,1)$, $a=1,\ldots,K$, $i=1,\ldots,N$.
 As in the
 previous section we will set $C$ in order to optimize the approximation. Let us call $\E_a$ the average with respect all the random variables $J_i^a$, $i=1,\ldots
 ,N$ and $\E_0$ the average with respect all the $J_{ij}$. We denote with $\E$ the average with respect to all $J$. Now we define recursively the random variables
 \be
 Z_K=\tilde Z_N; \ Z_{K-1}=(\E_KZ_K^{m_k})^{\frac{1}{m_k}};\ \ldots Z_0=(\E_1Z_1^{m_1})^{\frac{1}{m_1}},
 \ee
 where each $Z_a$ depends only on the external noise $J_{ij}$ and on the $J_i^b$ for $b\leq a$. Finally we define the auxiliary interpolating function
 \be\label{RSB2}
 \varphi_N(t;x(q))=\frac{1}{N}\E_0\log Z_0(t;x(q)),
 \ee
 that is completely averaged out with respect of all the external noises. Notice that, at $t=1$, we recover the original $A_N(\beta,\lambda)$, while, at
 $t=0$, we have a solvable one body interaction problem. Thus, we have the possibility to find an other sum rule for the free energy
 \be\label{sumrul}
 A_N(\beta,\lambda)=\varphi_N(t=0)+\int_0^1dt\frac{d}{dt}\varphi_N(t),
 \ee
 after calculating the $t$-derivative of $\varphi_N(t,x(q))$. For this purpose we need some additional definitios. Let us introduce the random variables
 \be \label{deff}
 f_a=\frac{Z_a^{m_a}}{\E_aZ_a^{m_a}}\ \ \ \ a=1,\ldots,K
 \ee
 and notice that they depend only on the $J^b_i$ for $b\leq a$ and they are normalized, $\E f_a=1$. Moreover we consider the t-dependent state $\omega$ associated
 to the \textit{Boltzmannfaktor} defined in $(\ref{RSB1})$ and its replicated $\Omega$. A very important rule is played by the following states $\tilde{\omega}_a$,
 with $a=1,\ldots, K$, and its replicated $\tilde{\Omega}_a$, defined as
 \be
 \tilde{\omega}_K(.)=\omega(.); \ \ \ \tilde{\omega}_a=\E_{a+1}\ldots\E_K(f_{a+1}\ldots f_K\omega(.)).
 \ee
 Finally we define the generalized $\meanv{.}_a$ average as
 \be
 \meanv{.}_a=\E(f_1\ldots f_a\tilde{\Omega}_a(.)).
 \ee
 The basic motivation for the introduction of an interpolating function like $\varphi(t;x(q))$ and the reason cause we tell about a broken replica symmetry
 bound, is the following
 \begin{teorema}\label{RSB3}
 The $t$-derivative of $\varphi_N(t)$, defined in $(\ref{RSB2})$, is given by
 \bea
 \frac{d}{dt}\varphi_N(t)&=&\frac{\beta^2}{4}\sum_{a=1}^K(m_{a+1}-m_a)q_a^2\nn\\
 &-&\frac{\beta^2}{4}\sum_{a=1}^K(m_{a+1}-m_a)\meanv{(q_{12}-q_a)^2}_a,
 \eea
 if we set the value $C=-\b^2\sum_{a=1}^K (q_a-q_{a-1})=-\beta^2Q$.
 \end{teorema}
 \begin{teorema}\label{RSB4}
 In the thermodynamic limit, for every  functional order parameter $x(q)$ of the type $(\ref{ordpar})$, the following sum rule holds
 \bea
 A(\beta,\lambda)&=&\varphi(0;x)+\frac{\beta^2}{4}\sum_{a=1}^K(m_{a+1}-m_a)q_a^2\nn\\
 &-&\frac{\beta^2}{4}\sum_{a=1}^K(m_{a+1}-m_a)\int_0^1\meanv{(q_{12}-q_a)^2}_adt
 \eea
 and, consequently, we have the following bound for the free energy density:
 \be
 -\beta f(\beta,\lambda)=A(\beta,\lambda)\leq\varphi(0;x)+\frac{\beta^2}{4}\sum_{a=1}^K(m_{a+1}-m_a)q_a^2.
 \ee
 \end{teorema}
 \nin Clearly, Theorem \ref{RSB4} follows from Theorem \ref{RSB3} by taking into account $(\ref{sumrul})$ and noting that the error term, containing overlap's
 fluctuation around every $q_a$, is negative defined.
 
 \nin Now let us go to Theorem \ref{RSB3}. The proof is straightforward and we will indicate only the main points. We begin with
 \begin{lemma}\label{rsbl1}
 $$\frac{d}{dt}\varphi(t;x)=\frac{1}{N}\E(f_1\ldots f_KZ_K^{-1}\partial_tZ_K)$$
 where
 \bea
 Z_K^{-1}\partial_tZ_K&=&\tilde{Z}_N^{-1}\partial_t\tilde{Z}_N\nn\\
 &=&\frac{\beta}{2\sqrt{2tN}}\sum_{i,j=1}^NJ_{ij}\omega(z_iz_j)-\frac{\beta^2}{4N}\omega((\sum_{i=1}^Nz_i^2)^2)\nn\\
 &-&\frac{\beta}{2\sqrt{1-t}}\sum_{a=1}^K\sqrt{q_a-q_{a-1}}\sum_{i=1}^NJ^a_i \omega(z_i)-\frac{C}{2}\sum_{i=1}^N\omega(z_i^2)\nn
 \eea
 \end{lemma}
 \begin{proof}
 From the definition $(\ref{deff})$ we get , for $a=0,1,\ldots K-1$
 $$
 Z_a^{-1}\partial_tZ_a=\E_{a+1}(f_{a+1}Z_{a+1}^{-1}\partial_tZ_{a+1})
 $$
 and the proof follows iterating this formula.
 \end{proof}
 \nin Now, using a standard integration by parts on the external noise, we get
 \bea
 \E(J_{ij}f_1\ldots f_K\omega(z_iz_j))&=&\sum_{a=1}^K\E(f_1\ldots \partial_{J_{ij}}f_a\ldots f_K\omega(z_iz_j))\nn\\
 &+&\E(f_1\ldots f_K\partial_{J_{ij}}\omega(z_iz_j))\nn
 \eea
 \bea
 \E(J_i^af_1\ldots f_K\omega(z_i))&=&\sum_{b=1}^K\E(f_1\ldots \partial_{J_i^a}f_b\ldots f_K\omega(z_i))\nn\\
 &+&\E(f_1\ldots f_K\partial_{J_i^a}\omega(z_i))
 \eea
 that can be completely evaluated using the following
 \begin{lemma}\label{rsbl2}
 For the $J$-derivative we have
 \bea
 \partial_{J_{ij}}\omega(z_iz_j)&=&\beta\sqrt{\frac{t}{2N}}(\omega(z_i^2z_j^2)-\omega^2(z_iz_j)),\label{l1}\\
 \partial_{J_i^a}\omega(z_i)&=&\beta\sqrt{1-t}\sqrt{q_a-q_{a-1}}(\omega(z_i^2)-\omega^2(z_i)),\label{l2}\\
 \partial_{J_{ij}}f_a&=&\beta\sqrt{\frac{t}{2N}}m_af_a(\tilde{\omega}_a(z_iz_j)-\tilde{\omega}_{a-1}(z_iz_j)),\label{l3}\\
 \partial_{J_i^a}f_b&=&0,\ \ \ \hbox{if}\  b<a,\label{l4}\\
 \partial_{J_i^a}f_b&=&\beta\sqrt{1-t}\sqrt{q_a-q_{a-1}} m_af_a\tilde{\omega}_a(z_i),\ \ \hbox{if}\  b=a,\label{l5}\\
 \partial_{J_i^a}f_b&=&\beta\sqrt{1-t}\sqrt{q_a-q_{a-1}} m_bf_b(\tilde{\omega}_b(z_i)-\tilde{\omega}_{b-1}(z_i)),\hbox{if}\  b>a.\label{l6}
 \eea
 \end{lemma}
 \begin{proof}
 Eq. $(\ref{l1})$, $(\ref{l2})$ follow from standard calculations, while Eq. $(\ref{l3})$ comes from the definition $(\ref{deff})$ and the easily established
 \bea
 \partial_{J_{ij}}Z_a^{m_a}&=&m_aZ_a^{m_a}Z_a^{-1}\partial_{J_{ij}}Z_a\nn\\
 Z_a^{-1}\partial_{J_{ij}}Z_a&=&\E_{a+1}(f_{a+1}Z_{a+1}^{-1}\partial_{J_{ij}}Z_{a+1})\ \ \hbox{for}\ a=0,\ldots,K-1\nn\\
 Z_K^{-1}\partial_{J_{ij}}Z_K&=&\beta\sqrt{\frac{t}{2N}}\omega(z_iz_j)\nn\\
 Z_a^{-1}\partial_{J_{ij}}Z_a&=&\beta\sqrt{\frac{t}{2N}}\E_{a+1}(f_{a+1}\ldots f_K\omega(z_iz_j))=\beta\sqrt{\frac{t}{2N}}\tilde{\omega}_a(z_iz_j)\nn
 \eea
 \nin In the same way we get Eq. $(\ref{l5})$, $(\ref{l6})$, where we have to remember that $f_b$ does not depend on $J^a_i$ if $b<a$.
 \end{proof}
 \nin If we use Lemma \ref{rsbl1} and Lemma \ref{rsbl2}, after some straightforward calculations, we obtain
 \bea
 \frac{d}{dt}\varphi_N(t)&=&\frac{\beta^2}{4}\sum_{a=1}^K(m_{a+1}-m_a)q_a^2\nn\\
 &-&\frac{\beta^2}{4}\sum_{a=1}^K(m_{a+1}-m_a)\meanv{(q_{12}-q_a)^2}_a\nn\\
 &+&\frac{1}{2N}(-\beta^2Q-C)\sum_{i=1}^N\meanv{z_i^2}_K
 \eea
 and we complete the proof of Theorem \ref{RSB3} setting $C=-\beta^2Q$.

\nin Now we should find a general expression for $\varphi_N(0;x)$, as in the following
\begin{teorema}
For any choice of the piecewise functional order parameter $x$, the initial condition $\varphi_N(0;x)$ is given by
\be
\varphi_N(0;x)=\log \s(Q)+f(0,0;x),
\ee
where $f(q,y;x)$ is the solution of the Parisi equation, i.e. the nonlinear antiparabolic partial differential equation
\be\label{eqp}
\partial_q f(q,y)+ {1\over2}\bigl(f^{\prime\prime}(q,y)+x(q){f^\prime}^2(q,y)\bigr)=0,
\ee
with final condition at $q=Q$
\be\label{eqpfc}
f(Q,y)=\frac{\b^2}{2}\s^2(Q)y^2.
\ee
\end{teorema}
\begin{proof}
Since the \textit{Boltzmannfaktor} factorizes at $t=0$, we have that
\bea
\tilde Z_N(0;x)&=&
\E_z \exp\left(\frac{(\l-\b^2Q)}{2}\sum_{i=1}^Nz_i^2\right)\exp\left(\b\sum_{a=1}^K\sqrt{q_a-q_{a-1}}\sum_{i=1}^NJ^a_iz_i\right) \nn\\
&=&\prod_{i=1}^N \s(Q)\exp\left(\frac{\b^2}{2}\s(Q)^2(\sum_{a=1}^K\sqrt{q_a-q_{a-1}}J^a_i)^2\right)\nn\\
&\equiv&\prod_{i=1}^N \s(Q)\exp\left(f(Q,\sum_{a=1}^K\sqrt{q_a-q_{a-1}}J^a_i)\right).\label{ersb1}
\eea
From the definition (\ref{RSB2}) of the interpolating function $\varphi_N(t;x)$, we note that, due to the $1/N$ factor, we can evaluate the (\ref{ersb1}) on a single site only. The $\s(Q)$ goes to form the $\log \s(Q)$ term and what remains is just the solution of the Parisi equation, evaluated at $y=0$, and propagated from $q=Q$ to $q=0$ through a series of gaussian integration as in \cite{RSB}.
\end{proof}

\nin Now we can exactly solve equation (\ref{eqp}) with final condition (\ref{eqpfc}) to find $f(0,0;x)$ and so $\varphi_N(0;x)$. Infact we give the following
\begin{lemma}
For any functional order parameter $x\in\mathcal{X}$, the solution of equation (\ref{eqp}) with final condition (\ref{eqpfc}), evaluated at $y=0$ and $q=0$ is given by
\be
f(0,0;x)=\frac 1 2 \b^2 \s^2(Q)\int_0^Q\frac{dq}{1-\b^2\s^2(Q)\int_q^Q x(q')dq'}
\ee
\end{lemma}
\begin{proof}
We look for a solution of (\ref{eqp}) of the form $f(q,y)=a(q)+\frac 1 2 b(q)y^2$. Since $f$ must fulfill final condition (\ref{eqpfc}), it has to be $a(Q)=0$ and $b(Q)=\b^2\s^2(Q)$. If we want $f(q,y)$ to be a solution of  (\ref{eqp})
\bea
&&\partial_q f(q,y)+ {1\over2}\bigl(f^{\prime\prime}(q,y)+x(q){f^\prime}^2(q,y)\bigr)\nn\\
&=&a'(q)+\frac 1 2 b(q) +\frac 1 2 y^2\left( b'(q)+x(q)b^2(q)\right)=0,
\eea
i.e. $f(q,y)$ is a solution of (\ref{eqp}) if $a(q)$ and $b(q)$ are solutions of the  ordinary differential equation's system
\bea
a'(q)+\frac 1 2 b(q)&=&0 \label{eqs1}\\
b'(q)+x(q)b^2(q)&=&0\label{eqs2},
\eea
with final conditions $a(Q)=0$ and $b(Q)=\b^2\s^2(Q)$. Integrating equation (\ref{eqs2}) we obtain
\be\label{eqs3}
\frac{1}{b(q)}=\frac{1}{\b^2\s^2(Q)}-\int_q^Q x(q')dq'.
\ee
Putting (\ref{eqs3}) into equation \ref{eqs1} and integrating, we have the proof.
\end{proof}
\nin Finally, from the continuity of $f(q,y;x)$ respect to the choice of the functional order parameter $x$ (see \cite{RSB}, \cite{leshouches}) and noticing that
\be
\frac{\beta^2}{4}\sum_{a=1}^K(m_{a+1}-m_a)q_a^2=\frac{\b^2}{4}Q^2-\frac{\beta^2}{2}\int_0^Q q x(q) dq,
\ee
and we can use \ref{RSB4} for stating our main result

\begin{teorema}
The pressure of the model is defined by the following variational principle:
 \be
 A(\beta,\lambda)=\inf_{x\in\mathcal{X}}\hat{A}(\beta,\lambda;x),
 \ee
 with
 \bea\label{defarsb}
 \hat{A}(\beta,\lambda;x)&=&
 \log \s(Q)+\frac 1 2 \b^2 \s^2(Q)\int_0^Q\frac{dq}{1-\b^2\s^2(Q)\int_q^Q x(q')dq'}\nn\\
 &+&\frac{\b^2}{4}Q^2-\frac{\beta^2}{2}\int_0^Q q x(q) dq,
 \eea
 Moreover the infimum is attaiened at the RS functional order parameter $x=0$, $0\leq q<q^{RS}$, $x=1$ elsewhere.
 
\end{teorema}
\begin{proof}
\nin The first part of the theorem is a direct conseguence of all the results in this section; so we will focus our attention only on its last part, that is
$$
A^{RS}(\beta,\lambda)\leq\inf_{x\in\mathcal{X}}\hat{A}(\beta,\lambda;x).
$$
\nin Let $x_{\epsi}$ a family of functional order parameter parametrized by $\epsi$ and consider
$\hat A(\b,\l;x_{\epsi})$. We will find the infimum of $\hat A(\b,\l;x)$ imposing that $\frac{d}{d\epsi}\hat A(\b,\l;x_{\epsi})|_{\epsi=0}=0$ for any family $x_{\epsi}$ passing through $x_0=x$ when $\epsi=0$. Using (\ref{defarsb}) and defining $\eta(q)=\frac{d}{d\epsi}x_{\epsi}(q)|_{\epsi=0}$, the infimum is achieved in $x$ satisfying
\bea
\frac{d}{d\epsi}\hat A(\b,\l;x_{\epsi})|_{\epsi=0}&=&-\frac{\b^2}{2}\int_0^Q q \eta(q) dq
+\frac 1 2 \b^4\s^4(Q)\int_0^Q dq\int_q^Q dq' \frac{\eta(q')}{\bar b^2 (q)}\nn\\
&=&-\frac{\b^2}{2}\int_0^Q q\  \eta(q) dq
+\frac 1 2 \b^4\s^4(Q)\int_0^Q dq\  \eta(q)\int_0^q  \frac{dq'}{\bar b^2 (q')}\nn\\
&=&-\frac{\b^2}{2}\int_0^Qdq\  \eta(q)\left(q-\b^2\s^4(Q)\int_0^q\frac{dq'}{\bar b^2 (q')}\right)=0\nn,
\eea
where we have defined $\bar b(q)=1-\b^2\s^2\int_q^Q x(q')dq'$. From the arbitrariness of $\eta(q)$, it has to be
\be
q-\b^2\s^4(Q)\int_0^q\frac{dq'}{\bar b^2 (q')}=\int_0^q dq'\left(1-\frac{\b^2\s^4(Q)}{\bar b^2 (q')}\right)=0,
\ee
for every $q$ and consequently $1-\frac{\b^2\s^4(Q)}{\bar b^2 (q')}$=0, i.e., recalling the definition of $\bar b(q)$,
\be
1-\b^2\s^2(Q)\int_q^Q x(q')dq'=\b\s^2(Q),
\ee
that has the solution $x=0$ and $Q=q^{RS}$, such that $\b\s^2(q^{RS})=1$.
\end{proof}

\nin So we have completed the RSB scheme, showing that the RSB bound for the free energy gives the same result of the RS approssimation. Anyway we have not yet proven that the replica symmetric solution is the exact infinite volume free energy of the model. That will be the subject of our last section.

\section{The Reverse Bound}

In this section we will show a lower bound for the gaussian pressure by using a well known method in statistical mechanics. The idea, coming from the proof of the equivalence between microcanonical and canonical ensemble \cite{Ru69}, is to cut the space into spherical shells of thickness $\e$ in such a way the integral over the whole space in the partition function is just the sum over all the shell; taking just one single shell then, we obtain a lower bound. 

\nin Namely, we fix $R\in[0,\infty)$ and define
\be 
S^{\e}_{R_N}=\{z\in\R^N: R\sqrt{N}-\e \leq \| z\|\leq R\sqrt{N} \},
\ee  
such that,
\bea 
Z^g_N(\b,\l,J)&=&\int_{\R^N}dz\frac{e^{-\|z\|^2/2}}{(2\pi)^{N/2}}e^{\left(-\beta H_N(z,J)-\frac{\beta^2}{4N}\|z\|^4+\frac{\lambda}{2}\|z\|^2\right)}\nn\\
&\geq &  \int_{S^{\e}_{R_N}}dz\frac{e^{-\|z\|^2/2}}{(2\pi)^{N/2}}e^{\left(-\beta H_N(z,J)-\frac{\beta^2}{4N}\|z\|^4+\frac{\lambda}{2}\|z\|^2\right)}\nn\\
&=& \frac{\e S_{R_N}}{(2\pi)^{\frac N 2}}\ \ \frac{1}{\e}\int_{S^{\e}_{R_N}}\frac{dz}{S_{R_N}}e^{\left(-\beta H_N(z,J)-\frac{\beta^2}{4N}\|z\|^4+\frac{(\lambda-1)}{2}\|z\|^2\right)}\nn\\
&\geq & \frac{\e S_{R_N}}{(2\pi)^{\frac N 2}}\ \ e^{-\frac{\b^2 R^4 N}{4}+\frac{(\l -1)R^2 N}{2}+ \Or{\e\sqrt{N}}}\ \  \frac{1}{\e}\int_{S^{\e}_{R_N}}\frac{dz}{S_{R_N}}e^{-\beta H_N(z,J)}\nn
\eea
Taking $\frac 1 N \E \log$ we get
\be 
A^g_N(\b,\l)\geq \frac 1 N\log(\frac{S_{R_N}}{(2\pi)^{\frac N 2}})-\frac{\b^2 R^4}{4}+\frac{(\l -1)R^2}{2} +A_{N,\e}^{sh}(\b,R_N)+ \Or{\frac 1{\sqrt{N}}}.
\ee
Since we can exchange the limits $\e\to0$ and $N\to\infty$ (as it is easy to prove), taking the infinite volume limit of both sides we obtain
\be 
A^g(\b,\l)\geq A^{sf}(\b,R)-\frac{\b^2 R^4}{4}+\frac{(\l -1)R^2}{2}+ \log(R)+\frac 1 2,
\ee
where an easy computation show that $\frac 1 N\log(\frac{S_{R_N}}{(2\pi)^{\frac N 2}})$ tends to $\log(R)+1/2$. Now we can take the supremum over $R\in[0,\infty]$ obtaining
\be\label{eq:main-formula} 
A^g(\b,\l)\geq\sup_{R\in(0,\infty)}\left[A^{sf}(\b,R)-\frac{\b^2R^4}{4}+\frac{(\l-1)R^2}{2}+\log R+\frac{1}{2} \right].
\ee

\nin In what follows we will see that the right side of ($\ref{eq:main-formula}$) is just the replica symmetric approximation $A^{g}_{RS}$ and thus we have just the lower bound we need.

\nin As it is well known \cite{KT}\cite{CS}\cite{PT}, the free energy of the spherical model can be expressed as the variational principle
$$
A^{sf}(\b,R=1)=\min_{q\in[0,1]}\frac{1}{2}\left( \frac{q}{1-q}+\log(1-q)+\frac{\b^2}{2}(1-q^2)\right),
$$
where the minimum is achieved at $q=0$, for $\b<1$, and at $q=1-\frac{1}{\b}$ otherwise. Furthermore, it's easy to check (through a change of variable) that, in the spherical model, a shift on the radius is equivalent to a temperature rescaling, \textit{i.e.}
\be 
A^{sf}(\b,R)=A^{sf}(\b R^2,1)=
\begin{cases}
\frac{\b^2 R^4}{4},  & \text{$\b R^2<1$}\\
\b R^2 -\log R-\frac 1 2 \log \b -\frac 3 4. & \text{$\b R^2>1$}.
\end{cases}
\ee

\nin First, we note that the function we want to to maximize
\be 
f(\b,R)=\begin{cases}
\frac{(\l-1)R^2}{2}+\frac 1 2 \log(R^2)+\frac 1 2,  & \text{$\b R^2<1$}\\
\b R^2 -\frac{\b^2R^4}{4}+\frac{(\l-1)R^2}{2}-\frac 1 2 \log \b -\frac 1 4, & \text{$\b R^2>1$}
\end{cases}
\ee
\nin is continuous in $R\in(0,\infty)$, goes to $- \infty$ at the interval's extremes and is concave in $R^2$ such that it must have a finite unique maximum. It is more useful to extremalize in $R^2$, such that
\be 
\frac{\partial}{\partial R^2}f(\b,R)=\begin{cases}
(\l-1)+\frac 1 {2R^2},  & \text{$\b R^2<1$}\\
\b -\frac{\b^2R^2}{2}+(\l-1), & \text{$\b R^2>1$}
\end{cases}
\ee
and
\be 
\frac{\partial}{\partial R^2}f(\b,R)=0 \Leftrightarrow \begin{cases}
R^2=\frac{1}{1-\l},  & \text{$\b R^2<1$}\\
R^2=\frac{2\b+(\l-1)}{\b^2}, & \text{$\b R^2>1$}
\end{cases}
\ee
Thanks to the concavity of $f$, for each value of $(\b,\l)$, only one critical point $\bar{R}^2$ can exists: if $\b<(1-\l)$, $\bar{R}^2=1/(1-\l)$, otherwise $\bar{R}^2=2\b+(\l-1)/\b^2$ and we obtain
\be \label{gaussian free energy}
A^g(\b,\l)\geq\begin{cases}
 -\frac 1 2 \log(1-\l)   & \text{$\b<1-\l$}\\
 -\frac 1 2 \log(\b)+\frac{\b \bar{q}}{2}+\frac{\b^2\bar{q}^2}{4},   & \text{$\b>1-\l$}
\end{cases}
\ee 
with $\bar{q}(\b,\l)=\frac{\b-(1-\l)}{\b^2}$. Right side of equation ($\ref{gaussian free energy}$) is exactly the replica symmetric approximation of the model. We can put all information we have found in the following final 

\begin{teorema}\label{MAIN} In the thermodynamic limit, the pressure of the gaussian spin glass model satysfies the following inequality
\begin{itemize}
\item $A^g(\b,\l)\leq A^g_{RS}(\b,\l)$ \nn\\
\item $A^g(\b,\l)\geq \sup_{R\in(0,\infty)}\left[A^{sf}(\b,R)-\frac{\b^2R^4}{4}+\frac{(\l-1)R^2}{2}+\log R+\frac{1}{2} \right]$. \nn
\end{itemize}
Moreover,
$$
\sup_{R\in(0,\infty)}\left[A^{sf}(\b,R)-\frac{\b^2R^4}{4}+\frac{(\l-1)R^2}{2}+\log R+\frac{1}{2} \right]=A^g_{RS}(\b,\l)
$$
thus
$$
A^g(\b,\l)=A^g_{RS}(\b,\l)=\sup_{R\in(0,\infty)}\left[A^{sf}(\b,R)-\frac{\b^2R^4}{4}+\frac{(\l-1)R^2}{2}+\log R+\frac{1}{2} \right].
$$
\end{teorema}

\section{Conclusions and Outlooks}

In this paper we introduced and solved a model of Gaussian spin glass. We have shown how to regularize the soft spins in order to tackle a right control of the thermodynamic observables and extend the tecniques developed for the Sherrington-Kirkpatrick model. 

\nin In this way we have achieved existence of thermodynamic limit of the free energy; we have computed the annealed free energy and, both looking at its own fluctuations, both looking at the fluctuations of the order parameter, we have found its region of validity in the $(\beta,\l)$ plane, so finding the critical line of the model.

\nin Furthermore, we have studied the Replica Symmetric approximation, and, by a deeper analysis through the RSB scheme, we have found that it actually give the same bound to the free energy of the model of the RS one. By further analysis we have shown that RS solution is infact exact.

\nin This was in a sense aspected, essentially for two reasons: the supposed relation with the Hopfield Model and the intimate connection with the Spherical Spin Glass \cite{CS}\cite{Tsfer}\cite{PT}\cite{T}. On the first one, that is our main matter of investigation, and also the reason why we need to introduce this model \cite{NN?}, we have already spent some words in the Introduction. About the second, there is a clear \textit{a priori} hint given by concentration of gaussian measure argoument that the two models are equivalent, strengthened by the \textit{a posteriori} fact that they share the same structure of the distribution of the overlap (both replica symmetric). This feature is not surprising, and was pointed out also in the final step of our proof, in the previous section. The relationship between the two can be deepen more, studing model with even more general noise. We plan to report soon about this topic.

\bigskip
\nin {\bf Acknowledgements\\}
\nin
The author are pleased to thank Andrea Crisanti and Michel Talagrand for useful conversations.
\newline
This work belongs to the strategy of investigation funded by the MIUR trough the FIRB grant RBFR08EKEV which is acknowledged together with a contribute by Sapienza Universit'a di Roma.
\newline
FG is further grateful to INFN (Istituto Nazionale di Fisica Nucleare).
\newline
AB is further grateful to GNFM (Gruppo Nazionale per la Fisica Matematica).
\newline
GG would like to thank Francis Comets,
Giambattista Giacomin and all the members of LPMA Universit\`e Denis Diderot Paris 7, for kind hospitality and
fruitful discussions.


\begin{thebibliography}{9}

\bibitem{Bip} A. Barra, G. Genovese, F. Guerra, \textit{Equilibrium statistical mechanics of bipartite spin systems}, J. Phys. A: Math. Theor. \textbf{44}, 245002 (2011).
\bibitem{NN} A. Barra, G. Genovese, F. Guerra, \textit{The replica symmetric behavior of the analogical neural network}, J. Stat. Phys. 140, 784 (2010).
\bibitem{NN?} A. Barra, G. Genovese, F. Guerra, D. Tantari, \textit{How glassy is neural network?}, preprint (2012).
\bibitem{HT} A. Barra, F. Guerra, \textit{Ergodicity in the analogical Hopfield model for neural networks}, J. Math. Phys. 49, 125127 (2008).
\bibitem{BDG} G. Ben Arous, A. Dembo, A. Guionnet, {\em Aging of spherical spin glasses}, Prob. Theor. Related Fields 120, 1, (2001).
\bibitem{BK} T. H. Berlin, M. Kac, {\em The Spherical Model of a Ferromagnet}, Phys. Rev. 86, 821 (1952).
\bibitem{BoK} A. Bovier, A. Klymovskiy,{\em The Aizenman-Sims-Starr and Guerra's schemes for the SK model with multidimensional
spins}, Electron. J. Probab. 14,161 (2009).
\bibitem{CS} A. Crisanti, H. -J. Sommers, {\em The spherical p-spin interaction p-spin glass model: the statistics}, Z. Phys. B Condensed Matter 83, 341, (1992).
\bibitem{io1} G. Genovese, A. Barra, {\em A mechanical approach to mean field spin models}, J. Math. Phys. \textbf{50}, 365234 (2009).
\bibitem{FS} Y. V. Fyodorov, H. -J- Sommers {\em Classical Particle in a Box with Random Potential: exploiting rotational symmetry of replicated hamiltonian}, Nuclear Physics B 764,128 (2007).
\bibitem{sum-rules} F. Guerra, {\em Sum rules for the free energy in the mean field spin glass model}, in {\em Mathematical Physics in Mathematics and
Physics: Quantum and Operator Algebraic Aspects}, Fields Institute Communications {\bf 30}, 161 (2001).
\bibitem{leshouches} F. Guerra, {\em An Introduction to Mean Field Spin Glass Theory: Methods and Results}, A. Bovier et al. eds, Les Houches, Session LXXXIII, (2005).
\bibitem{RSB} F.Guerra, \textit{Broken Replica Symmetry Bounds in the Mean Field Spin Glass Model}, Comm. Math. Phys. \textbf{233}, 1 (2003).
\bibitem{GT} F. Guerra, F. L. Toninelli, \textit{The Thermodynamic Limit in Mean Field Spin Glass Models, Comm. Math. Phys.} \textbf{230}, 71 (2002).
\bibitem{GTQ} F. Guerra, F. L. Toninelli, \textit{Quadratic replica coupling in the Sherrington Kirkpatrick mean field spin glass model, J. Math. Phys.} \textbf{43}, 3704 (2002).
\bibitem{GTF} F. Guerra, F. L. Toninelli, \textit{Central limit theorem for fluctuations in the high temperature region of the Sherrington Kirkpatrick spin glass model, J. Math. Phys.} \textbf{43}, 6224 (2002).
\bibitem{KT} J. M. Kosterlitz, D. J. Thouless, Raymund C. Jones, \textit{Spherical model of a spin glass}, Phys. Rev. Lett. 36, 1217 (1976).
\bibitem{PT} D. Panchenko, M. Talagrand {\em On the overlap in the multiple spherical models}, Ann. Probab. 35, 2321 (2007).
\bibitem{Ru69} D. Ruelle. {\em Statistical Mechanics. Rigorous results}, W.A. Benjamin Inc., New York, 1969.
\bibitem{Tsfer} M. Talagrand, \textit{Free energy of the spherical mean field model}, Prob. Theor. Related Fields 134, 339 (2006).
\bibitem{T} M. Talagrand, \textit{Spin glasses: a challenge for mathematicians. Cavity and mean field models}, Springer Verlag (2003).


\end{thebibliography}
\end{document}